\documentclass[11pt,onecolumn]{article}
\usepackage{amsmath, amssymb}
\usepackage{amsthm}
\usepackage{float}
\usepackage{enumitem}
\usepackage{color}
\usepackage{cleveref}

\usepackage{mleftright}

\usepackage[linesnumbered,ruled,vlined]{algorithm2e}
\usepackage[noend]{algpseudocode}
\usepackage{graphicx,tikz,subfig}
\usetikzlibrary{calc,fit}

\usepackage[left=0.9in,top=0.9in,right=0.9in,bottom=0.5in,nohead,footskip=0.3in]{geometry}

\usepackage{titlesec}
\titlespacing*{\section}{0pt}{11pt}{11pt}
\titlespacing*{\subsection}{0pt}{11pt}{11pt}

\newtheorem*{rep@theorem}{\rep@title}
\newcommand{\newreptheorem}[2]{%
\newenvironment{rep#1}[1]
{ \def\rep@title{#2 \ref{##1}} \begin{rep@theorem}} {\end{rep@theorem}} }

\newtheorem{theorem}{Theorem}
\newreptheorem{theorem}{Theorem}
\newtheorem{lemma}[theorem]{Lemma}

\newtheorem{observation}[theorem]{Observation}
\newtheorem{proposition}[theorem]{Proposition}
\newtheorem{corollary}[theorem]{Corollary}

\theoremstyle{definition}
\newtheorem{definition}[theorem]{Definition}

\newcommand{\E}{{\mathbb E}}

\begin{document}

\title{A note on rounding fractional matchings with \\ constant-factor strong negative correlation}
\author{David G. Harris\footnote{University of Maryland. Email: {\tt davidgharris29@gmail.com}}}

\maketitle

\abstract{We describe new dependent-rounding algorithms for bipartite graphs. Given a fractional matching $x$ of graph $G = (U \cup V, E)$, the algorithms return an integral solution $X$ such that each right-node $v \in V$ has exactly one edge, and where the variables $X_e$ also satisfy broad non-positive correlation properties. In particular, for distinct edges $e, f$ sharing a left-node $u \in U$, the variables $X_{e}, X_{f}$ have \emph{strong} negative-correlation, i.e. the expectation of $X_{e} X_{f}$ is significantly below $x_{e} x_{f}$. 

Dependent rounding schemes with these properties have been used for  approximation algorithms for job-scheduling on unrelated machines to minimize weighted completion times, among other applications. Our new algorithm achieves simpler and qualitatively stronger bounds compared to prior algorithms. In particular, we achieve a negative-correlation property $$
\E[X_{e} X_{f}] \leq 0.79751 \ x_{e} x_{f},
$$
which is a significant constant-factor improvement over Baveja, Qu \& Srinivasan (2024). 

We show that the constant cannot be reduced below $41/80$ by any comparable rounding algorithm.
}

\section{Introduction}
Many scheduling and resource 
allocation problems can be formulated in terms of  a \emph{bipartite assignment problem}: we are given a complete bipartite graph $G = (U \cup V, E)$, and we wish to select a ``half-matching" $K$, i.e. a set of edges $K$ which intersects each right-node $v \in V$ exactly once. For instance, $V$ can represent a set of jobs to be scheduled, and $U$ can represent a set of possible machines. Alternatively, $V$ can represent items to be sold, and $U$ can represent potential buyers.

A popular strategy for such problems is to first solve a relaxation (e.g., a linear program), obtaining fractional vectors $x_e: e \in E$, and then round this to an integral solution $X_e: e \in E$ where $X_e$ is the indicator that $e \in K$. For this strategy to be viable, the fractional solution $x$ must satisfy $x(N(v)) = 1$ for all right-nodes $v$. (Here, $N(v)$ denotes the set of edges incident on vertex $v$ and we write $x(L) = \sum_{e \in L} x_e$ for any edge-set $L$). One particularly nice scenario has  the left-nodes also satisfying $x(N(u)) = 1$, i.e. a fractional matching. 

The simplest rounding method, known as \emph{independent rounding}, is that each right-node $v \in V$ independently selects one neighboring edge $e \in N(v)$, wherein each edge is selected with probability $x_e$. However, the left-nodes (e.g. the machines in a scheduling problem) can then become unevenly loaded due to random fluctuations. This can be undesirable for some allocation problems.

A new rounding approach was proposed in \cite{bansal2016lift}, based on \emph{dependent rounding with strong negative correlation.} This was applied to a classical scheduling problem of minimizing weighted completion time on unrelated machines. The edges are tied together in a scheme wherein each edge $e$ is still \emph{marginally} selected with probability $x_e$, while neighboring edge pairs satisfy a stronger property:
\begin{equation}
\label{ata66}
\E[X_e X_f] \leq c \ x_e x_f \qquad \text{for a  constant $c < 1$}
\end{equation}

Of course, independent rounding would give $\E[X_e X_f] = x_e x_f$ exactly.

 The guarantee of (\ref{ata66}) is known as \emph{(constant-factor) strong negative correlation.}
Since then, a variety of rounding schemes with strong negative correlation have been developed, leading to improved approximation ratios for various scheduling problems \cite{im2020,im2023,baveja2024,harris2024dependent,harris2025}. Generally speaking, these fall into two classes. The first, which includes the original work of \cite{bansal2016lift}, uses a random walk: at each stage, the fractional vector $x$ is modified, until eventually it becomes integral. The modification rule is chosen so that if two edges $e,f$ share a left-node, then one value $x_e$ is incremented and the other value $x_f$ is decremented. An improved version of this rounding scheme was later developed in \cite{baveja2024}.

The second class, developed originally by \cite{im2023}, is based on ideas from \emph{contention-resolution} in auctions. Here, each left-node ``bids" for the right-nodes, and these demands are balanced so that each left-node does not get too many edges.  The original work of \cite{im2023} was based on Poissonian tickets for the allocation; the later work  \cite{harris2024dependent} was based on using a multivariate geometric distribution. Most recently, the work \cite{harris2025} developed a method based on Dirichlet random variables; we will discuss this in greater detail later.

These two rounding strategies have quite different behaviors. The contention-resolution algorithms are typically most powerful where the entries of $\vec x$ are infinitesimal. For example, for a fractional matching $\vec x$, the algorithm of \cite{harris2025} achieves
$$
\E[ X_e X_f ] \leq \frac{1}{2} \cdot x_e x_f (1 + o(1))
$$
where the $o(1)$ term applies as $x_e, x_f \rightarrow 0$.  The contention-resolution algorithms appear to provide essentially no strong negative correlation for large values of $x_e, x_f$. 

By contrast, the random-walk based algorithms typically provide a constant-factor negative correlation; for example, the algorithm of \cite{baveja2024} gives
$$
\E[ X_e X_f ] \leq \frac{26}{27} \cdot x_e x_f
$$
for all values $x_e, x_f$. 

For algorithmic applications, the infinitesimal values are typically the bottleneck in determining the approximation ratio. Thus the random-walk algorithms end up being quantitatively weaker than the Dirichlet-based algorithm. However, the contention-resolution algorithms have extremely complicated and non-algebraic formulas; it is painful to carry out the numerical analysis and often requires large computer searches to bound the relevant formulas.

\subsection{Our contribution}
Our main contribution is to develop a new random-walk based rounding algorithm for negative correlation. It is a variant of the algorithm of \cite{baveja2024}. On its own, it already significantly improves over the previous negative correlation factor.

Going one step further, we observe that this algorithm has better negative correlation when the values $x_e, x_f$ are \emph{large},  precisely the opposite behavior of contention-resolution. We exploit this property via a hybrid approach which interpolates between the new rounding algorithm and the Dirichlet rounding algorithm. This gives even better constant-factor guarantees over all ranges.

In order to state our results in greatest generality, we introduce a key definition. 

\begin{definition}[Stable edge set \cite{harris2024dependent}]
An edge set $S \subseteq E$ of $G$ is \emph{stable} if it has no edges $e_1, e_2 \in S$ whose distance in the line graph of $G$ is precisely two. 
\end{definition}

We can formally state our rounding method as follows:
\begin{theorem}
\label{mainthm34}
Let  $\vec x$ be a fractional matching of the simple bipartite graph $G = (U \cup V, E)$. Our rounding algorithm generates an integral solution $\vec X \in \{0, 1 \}^E$ with the following properties:

\begin{itemize}
\item[(A1)] $\sum_{e \in N(v)} X_e = 1$ for all right-nodes $v$.
\item[(A2)] For each edge $e$ there holds $\E[X_e] = x_e$
\item[(A3)] For any stable edge-set $S \subseteq E$, there holds $
\E[ \prod_{e \in S} X_e ] \leq \prod_{e \in S} x_e
$
\item[(A4)] For any pair of distinct edges $e, f$ sharing a common left-node, we have $
\E[X_e X_f] \leq 0.79751 \ x_e x_f
$
\end{itemize}
\end{theorem}

This constant $0.79751$ is a significant improvement over the value $\frac{26}{27} \approx 0.96$ in the previous algorithm of \cite{baveja2024}. As we show in \Cref{lb}, a simple construction with the graph $G = C_{30}$ shows that the constant cannot be reduced below $41/80$, even if we only require \emph{pairwise} negative correlation instead of the full property (A3). In particular, the Dirichlet infinitesimal bound is strictly stronger in its regime compared to constant-factor strong-negative-correlation.

\bigskip

Before proceeding further, let us make a few comments on how to interpret the result.

First, note that if our goal was to achieve (A1), (A2), (A3) alone, then we could simply use independent rounding. However, this would give $\E[X_e X_f] \leq x_e x_f$, much weaker than (A4).

Alternatively, if our goal was to achieve (A1), (A2), (A4) alone, then we could round $\vec x$ to a matching of $G$. This is possible because $G$ is bipartite, so the fractional matching polytope does not have an integrality gap. This would give $\E[X_e X_f] = 0$ for edges $e,f$ sharing a left-node. However, this does not give general negative correlation for stable edge sets, or even for pairs of edges at distance $\geq 3$.

If we relaxed the restriction of $\vec x$ being a fractional matching, then it may be impossible to guarantee (A4) or any similar bound. For example, if we had a fractional vector $\vec x$ with $x_e = 1, x_f = 0.01$ for edges $e,f$ sharing a left-node, then property (A2) would ensure that $\E[X_e X_f] = x_e x_f$, without any strong negative correlation.

Finally, note that \cite{bansal2016lift} and \cite{baveja2024} considered a slightly different setting where the edges of $G$ are grouped into ``blocks".  Strong negative correlation should hold within a block, and negative correlation within a vertex. All our results also apply in this setting, via the transformation of replacing each block of the graph $G$ with a separate left-node in a new bipartite graph $G'$.

\section{Brownian-motion algorithm}
The first algorithm we consider is based on a Brownian motion in the fractional half-matching polytope.  It maintains a set of ``protected edges'' $P$ which obey the fractional matching constraint.  Here, $p \in [0,1]$ is a parameter to choose later.

We say an edge $e$ is \emph{fractional} if $x_e \in (0,1)$ strictly, else it is \emph{integral}.

\begin{algorithm}[H]
\caption{\sloppy {$\textsc{DepRoundBrownian}(\vec x,p)$}}
Each left-node $u$ draws random variables $R_u \sim \text{Unif}([0,1])$ and  $Y_u \sim \text{Bernoulli}(p)$. \\
Initialize edge-set $P$ by setting $P \leftarrow \{ (u,v) \in E : Y_u = 1 \}$ \\
\While{$P \neq \emptyset$} {
\If{there is any edge $e = (u,v) \in P$ with $x(e) \in \{0,1 \}$ or $|N(u) \cap P| = 1$}{ 
Update $P \leftarrow P \setminus \{e \}$}
\ElseIf{any right-node $v$ has $x( N(v) \cap P) = 1$}{
 Randomly select an edge $e \in N(v) \cap P$, where each edge is chosen with probability $x_e$. \\
   Update $P \leftarrow P \setminus (N(v) \setminus \{e \})$. \\
   }
   \Else{
   Choose a left-node $u$ with $N(u) \cap P \neq \emptyset$; \hspace{10in} \thickspace  \hspace{0.5in} if there are multiple such nodes, take the one with largest $R_u$. \\
   Choose any two edges $e_1 = (u, v_1), e_2 = (u,v_2) \in P$. \\
   Choose any two fractional edges $f_1 = (u_1, v_1) \in N(v_1) \setminus P,  f_2 = (u_2, v_2) \in N(v_2) \setminus P$.\\   
   Choose $\alpha > 0$ maximal such that $x_{e_1} \pm \alpha, x_{e_2} \pm \alpha, x_{f_1} \pm \alpha, x_{f_2} \pm \alpha$ are all in the range $[0,1]$. \\
   Draw random variable $Z$ uniformly from $\{-1,1\}$. \\
   Update $x_{e_1} \leftarrow x_{e_1} + Z \alpha, x_{e_2} \leftarrow x_{e_2} - Z \alpha, x_{f_1} \leftarrow x_{f_1} - Z \alpha, x_{f_2} \leftarrow x_{f_2} + Z \alpha$. \\
   }
   }
   Set $X = \textsc{IndependentRounding}(\vec x)$. \\
 \textbf{return} $X$
\end{algorithm}

 The overall algorithm is very similar to, and is inspired by, the algorithm of \cite{baveja2024}.
After generating $P$ and $R$, the algorithm successively updates $\vec x$ and removes edges from $P$ until $P = \emptyset$. Then a final integral value  is produced by independent rounding. We refer to the loop at Lines 3 -- 15 as the \emph{main loop.} Each step is polynomial time. 

Overall, we will show the following result about this algorithm:

\begin{theorem}
The rounding algorithm satisfies properties (A1), (A2), and (A3). 
\end{theorem}

\begin{theorem}
\label{gen-neg3}
For any pair of distinct edges $\bar e_1,\bar e_2$ sharing a left-node, there holds
$$
\E \bigl[ X_{\bar e_1} X_{\bar e_2} \bigr] \leq  x_{\bar e_1} x_{\bar e_2} \bigl( 1 -p + p^2 ( 2 -  x_{\bar e_1} - x_{\bar e_2})  - \min\{ p^3, p^2/2 \} (1 - x_{\bar e_1})(1 - x_{\bar e_2})  \bigr)
$$
\end{theorem}

The proof is lengthy and requires many intermediate steps.

\subsection{Basic analysis}
 To begin, we observe the following simple facts about the algorithm.

\begin{observation}
\label{obs0}
The following invariants hold at all intermediate stages:
\begin{enumerate}
\item $x(N(u) \cap P) \leq 1$ for all left-nodes $u$.
\item $x(N(v)) = 1$ for all right-nodes $v$.
\end{enumerate}
\end{observation}
\begin{proof}
The updates at Lines 4 -- 5 and 6 -- 8 do not change $x$, and they only remove edges from $P$.  The update at Lines 10 -- 15 has been carefully chosen to ensure that $x(N(u) \cap P)$ does not change, since any change to $e_1 \in N(u) \cap P$ is compensated by exactly the opposite change to $e_2 \in N(u) \cap P$. Likewise $x(N(v_1))$ does not change, since the change to $e_1 \in N(v_1)$ is compensated by the exactly opposite change to $f_1 \in N(v_1)$. Likewise $x(N(v_2))$ does not change.
\end{proof}

\begin{observation}
\label{obs01}
At Line 11, there exist choices for edges $e_1, e_2, f_1, f_2$ and $\alpha > 0$.
\end{observation}
\begin{proof}
Because the case at Line 4 was not activated, $u$ cannot have just a single edge in $P$, and there are no integral edges in $P$. Since $N(u) \cap P \neq \emptyset$, there must be available edges $e_1,e_2 \in P$, which are fractional. Since the case at Line 6 did not activate, we have $x(N(v_1) \cap P) < 1$ strictly; since $x(N(v_1)) = 1$, there must exist some fractional edge $f_1 \in N(v_1) \setminus P$. By a completely symmetric argument also some edge $f_2$ must be available. 

The choice $\alpha > 0$ is possible since all the edges $e_1, e_2, f_1, f_2$ are fractional.
\end{proof}

\begin{observation}
\label{obs1}
The expected number of iterations is $O( |E| )$; in particular, it terminates with probability one. The resulting value $X = \textsc{DepRoundBrownian}(x,p)$ satisfies (A1) and (A2).
\end{observation}
\begin{proof}
These properties were shown in \cite{baveja2024}, but we repeat them here for completeness.

In each application of Lines 4 -- 5 or Lines 6 -- 8, at least one edge is dropped from $P$; hence it can only be executed $|E|$ times. In each application of Lines 10 -- 15, there is at least one choice of $Z$ such that the updated vector $x'$ has one additional integral variable compared to $x$. Hence, for any value $t$, there are $O(1)$ rounds in expectation where $x$ has exactly $t$ integral variables. 

Property (A1) follows directly from the final independent rounding step while \Cref{obs0} ensures that the probabilities at each right-node sum to one. For property (A2), note that the modification step at Line 15 and the final independent rounding step are both unbiased.
\end{proof}

\begin{proposition}
\label{ya1}
Let $L$ be a stable edge-set. Conditional on the state $x$ at the beginning of any given iteration, the updated state $x'$ after the iteration satisfies
$$
\E \Bigl[ \prod_{e \in L} x'_e  \Bigr] \leq \prod_{e \in L} x_e
$$
\end{proposition}
\begin{proof}
If the current iteration does not execute Lines 10–15, then $x' = x$ and the claim is immediate. So suppose Lines 10–15 are executed, with modified edges $M_{+} = \{e_1, f_2 \}$ and $M_{-} = \{e_2, f_1 \}$ and $M = M_{-} \cup M_{+}$.  Note that $e_1, f_2$ have distance two (sharing neighbor $e_2$) and likewise $e_2, f_1$ have distance two. So, since $L$ is stable, we have $|L \cap M_{\pm}| \leq 1$. 

Now, if $M \cap L = \emptyset$, then $\prod_{e \in L} x'_e = \prod_{e \in L} x_e$ with probability one. If $|M \cap L| = 1$, say $M \cap L = \{g \}$, then we can write 
$$
\E \Bigl[ \prod_{e \in L} x'_e \Bigr] = \E[x'_g ]  \prod_{e \in L \setminus \{g \}} x_e
$$
where, because the update is unbiased, we have $\E[x'_g] = x_g$.

Finally, suppose that $|M \cap L| \geq 2$; this is only possible if $| M_{-} 
\cap L| = |M_{+} \cap L| = 1$.  Let us assume for concreteness that $M \cap L = \{e_1, e_2 \}$ (all possibilities are symmetric). Then, $\E[ x'_{e_1} x'_{e_2} ] = \E[ (x_{e_1} + \alpha Z)(x_{e_2} - \alpha Z) ] = x_{e_1} x_{e_2} + \alpha (x_{e_2} - x_{e_1}) \E[ Z ] - \alpha^2 \E[Z^2]$. Since $\E[Z] = 0$ and $Z^2 = 1$, this is $x_{e_1} x_{e_2} - \alpha^2$. So indeed
\[
\E \Bigl[ \prod_{e \in L } x'_e \Bigr] = (x_{e_1} x_{e_2} - \alpha^2) \prod_{e \in L \setminus \{e_1, e_2 \}} x_e \leq   x_{e_1} x_{e_2} \prod_{e \in L \setminus \{e_1, e_2 \}} x_e   = \prod_{e \in L} x_e. \qedhere
\]
\end{proof}

Property (A3) holds by \Cref{ya1} applied to each iteration and noting that $X$ is derived from $x$ by independent rounding.
We now turn to the harder strong-negative-correlation property.

\subsection{Proof of \Cref{gen-neg3}}
We will analyze a left-node $\bar u$ and a pair of incident edges $\bar e_1 = (\bar u, \bar v_1)$ and $\bar e_2 = (\bar u, \bar v_2)$.   Define
\begin{eqnarray*}
&\bar x_i = x_{\bar e_i}, \quad A_i =  N(\bar v_i) \setminus \{ \bar e_i \} \qquad i = 1,2
\end{eqnarray*}
and define $S$ to be the set of edge pairs $(h_1, h_2)$ with the following properties:
\begin{itemize}
\item $h_1 \in A_1 \cap P, h_2 \in A_2 \cap P$; and
\item Either (i) $h_1, h_2$ have distinct left-nodes, or

\hspace{0.38in} (ii) $h_1, h_2$ share a common left-node $u$ where $R_{u} < R_{\bar u}$.
\end{itemize}

Let us further define
$$
a_1 = x(A_1 \cap P), \quad a_2 = x(A_2 \cap P), \qquad s = \sum_{(h_1, h_2) \in S} x_{h_1} x_{h_2}
$$
where we note the inequality
\begin{equation}
\label{flup}
a_1 + a_2 - s \geq 
x( A_1 \cap P) + x(A_2 \cap P) - \sum_{\substack{h_1 \in A_1 \cap P,\\ h_2 \in A_2 \cap P}} x_{h_1} x_{h_2} \\
 = a_1 + a_2 - a_1 a_2 \geq 0
 \end{equation} 
 
\bigskip

 The key to the analysis is to track the following potential function:
$$
\Phi  = \begin{cases}
\bar x_1 \bar x_2 ( a_1 + a_2 - s ) & \text{if $\bar e_1, \bar e_2 \in P$} \\
\bar x_1 \bar x_2 & \text{otherwise}
\end{cases}
$$

\begin{lemma}
\label{gen-neg2}
Consider any iteration of the main loop; let $x',\Phi'$ be the updated values after the iteration. Conditional on the current state, there holds $\E[ \Phi'] \leq \Phi$.
\end{lemma}
\begin{proof}
Here we write $S', a_1', a_2', P'$ etc. for the updated values after the iteration.

If $\bar e_1 \notin P$ or $\bar e_2 \notin P$, then $\Phi = \bar x_1 \bar x_2$ and $\Phi' = \bar x_1' \bar x_2'$. In this case, the result follows from \Cref{ya1} (since $\{ \bar e_1, \bar e_2 \}$ is stable). So we suppose that $\bar e_1, \bar e_2 \in P$. Then there are a number of cases for the analysis.

\bigskip

\noindent \textbf{Case I: The update at Line 4 -- 5 runs for an edge $\pmb{e}$.} If $e = \bar e_1$ or $e = \bar e_2$, then by \Cref{obs0} we have $\bar x_1 \bar x_2 = 0$ and $\Phi' = \Phi = 0$. If $e \in A_1 \cap P$, then $a_1$ is decremented by $x(e)$ and $s$ is decremented by at most $a_2 x(e) \leq x(N(\bar v_2)) x_e \leq x_e$; in particular, $\Phi' \leq \Phi$. A completely analogous argument holds if $e \in A_2 \cap P$. If $e$ is any other edge then $\Phi' = \Phi$.

\bigskip

\noindent \textbf{Case II: The update at Line 6 -- 8 runs for a vertex $\pmb{v}$.} This will only modify the membership in $P$ for neighbors of $v$.  So if $v \notin \{ \bar v_1, \bar v_2 \}$, then $\Phi' = \Phi$. So suppose without loss of generality that $v = \bar v_1$, and so
$$
1 = x( N(\bar v_1) \cap P) = \bar x_1 + a_1
$$

With probability $\bar x_1$,  edge $\bar e_1$ gets selected to remain in $P$. In this case, $  A_1 \cap P' = \emptyset$ and $A_2 \cap P' = A_2 \cap P$ and $S' = \emptyset$, so $\Phi'=  \bar x_1 \bar x_2 a_2$. On the other hand, if $\bar e_1$ is not selected, then $\bar e_1 \notin P'$ and $\Phi' = \bar x_1 \bar x_2$. Putting the two cases together, we have
\begin{align*}
\E [\Phi' ] &= \bar x_1 \bigl(  \bar x_1 \bar x_2 a_2 \bigr) + (1 - \bar x_1)  \bigl( \bar x_1 \bar x_2 \bigr)  =  \bar x_1 \bar x_2 ( a_1 + a_2 - a_1 a_2 ) \\
&\leq  \bar x_1 \bar x_2 ( a_1 + a_2 - s  ) = \Phi
\end{align*} 
where the last inequality uses the bound (\ref{flup}).

\bigskip

\noindent \textbf{Case III: The update at Lines 10 -- 15 runs for a vertex $\pmb{u \neq \bar u}$.} 
Here $\bar u$ has two edges $\bar e_1, \bar e_2$ in $P$. So $\bar u$ would be eligible for the update. So, because of the tie-breaking rule, we must have \begin{equation}
\label{tat1}
R_u > R_{\bar u}.
\end{equation}

The only edges in $P$ whose value gets modified are $e_1, e_2$, so $\bar x'_1 = \bar x_1$ and $\bar x'_2 = \bar x_2$. The update steps are always unbiased, so $\E[ a'_1] = a_1$ and $\E[a_2'] = a_2$ and
$$
\E[s'] = \sum_{(h_1, h_2) \in S} \E[x'_{h_1} x'_{h_2} ]
$$

Now consider an edge pair $(h_1, h_2) \in S$. We cannot have $ |\{h_1, h_2 \} \cap \{e_1, e_2 \}| = 2$, as this would imply that $h_1, h_2, e_1, e_2$ all have the left-node $u$. But this is impossible since $S$ is defined to only include pairs sharing a left-node with \emph{smaller} value $R$ than $\bar u$, which contradicts (\ref{tat1}).

So $ |\{h_1, h_2 \} \cap \{e_1, e_2 \}| \leq 1$. Then since the update is unbiased we have $$
\E[x'_{h_1} x'_{h_2}] = x_{h_1} x_{h_2},
$$
and hence overall $\E[s'] = s$ and $
\E[ \Phi' ]  =  \bar x_1 \bar x_2 ( a_1 + a_2 - s) = \Phi.$

\bigskip

\noindent \textbf{Case IV: The update at Lines 10 -- 15 runs for vertex $\pmb{u = \bar u}$.} Again, $S$ does not change.  The only edges in $P$ whose value gets modified are $e_1, e_2$. Since such edges have left-node $\bar u$, they cannot be in $A_1$ or $A_2$.  Accordingly, we have 
$$
a_1' + a_2'  - s' =  a_1 +a_2- s
$$

So $\E[ \Phi' ] = \E[ \bar x_1' \bar x_2' ] ( a_1 + a_2 - s)$. By (\ref{flup}), the quantity $(a_1 + a_2 - s)$ is nonnegative. By \Cref{ya1}, we have $\E [\bar x'_1 \bar x'_2 ] \leq \bar x_1 \bar x_2$. So indeed
$\E[ \Phi' ] \leq \bar x_1 \bar x_2 (a_1 + a_2 - s) = \Phi$.
\end{proof}

\begin{proposition}
\label{aff1}
When drawing the random variables $Y$  and $R$, we have
$$
\E[ \Phi ] \leq \bar x_1 \bar x_2  \bigl( 1 -p + p^2 ( 2 -  \bar x_1 - \bar x_2)  - \min\{ p^3, p^2/2 \} (1 - \bar x_1)(1 - \bar x_2)  \bigr)
$$
\end{proposition}
\begin{proof}
Each edge $(u,v)$ goes into $P$ precisely when $Y_u = 1$, which holds with probability $p$. So
$$
\E[a_i] = \E[ x(A_i \cap P) ] = p x(A_i), \qquad i = 1,2
$$

Consider a pair of edges $h_1 \in A_1, h_2 \in A_2$. If $h_1, h_2$ have distinct left-nodes $u_1, u_2$, then $(h_1, h_2)$ goes into $S$ precisely when $Y_{u_1} = Y_{u_2} = 1$, which holds with probability $p^2$. Otherwise, if they share a common left-node $u$, then $(h_1, h_2)$ goes into $S$ precisely when $Y_u = 1$ and $R_u < R_{\bar u}$, which holds with probability $p/2$.  So
$$
\E [s] \geq \min\{p^2, p/2 \} x(A_1) x(A_2) 
$$

Furthermore, all these bounds are independent of $Y_{\bar u}$. Now, if $Y_{\bar u} = 0$, which holds with probability $1-p$, then $\Phi = \bar x_1 \bar x_2$; otherwise, if $Y_{\bar u} = 1$, then 
\begin{align*}
\E[ \Phi \mid Y_{\bar u} = 1]  &= \bar x_1 \bar x_2 \E \bigl[ a_1 + a_2 - s  \mid Y_{\bar u} = 1 \bigr] \leq \bar x_1 \bar x_2 \Bigl( p x(A_1) + p x(A_2) - \min \{p^2, p/2 \} x(A_1) x(A_2) \Bigr)
\end{align*}

Since $x$ is a fractional matching, $x(A_1) = 1 - \bar x_1$ and $x(A_2) = 1 - \bar x_2$. So:
$$
\E[ \Phi \mid Y_{\bar u} = 1] \leq \bar x_1 \bar x_2 \Bigl(  p (1-\bar x_1) + p(1-\bar x_2) - \min \{p^2, p/2 \} (1 - \bar x_1) (1 - \bar x_2) \Bigr).
$$

Putting the two cases together gives:
$$
\E[ \Phi ] \leq (1-p) \bar x_1 \bar x_2 + p \bar x_1 \bar x_2 \Bigl(  p (1 - \bar x_1) + p(1 - \bar x_2) - \min\{p^2, p/2 \} (1 - \bar x_1) (1 - \bar x_2) \Bigr),
$$
which is the claimed bound after some rearrangement.
\end{proof}

\begin{proof}[Proof of \Cref{gen-neg3}]
Let $ x', P', \Phi'$ denote the state after all iterations of the main loop. By  using \Cref{gen-neg2} over all loop iterations, we get
$$
\E[ \Phi' ] \leq \E[ \Phi ] 
$$
By \Cref{aff1}, the initial value $\Phi$ satisfies
$$
\E[ \Phi ] \leq (1-p) \bar x_1 \bar x_2 + p \bar x_1 \bar x_2  \bigl(  p (1 - \bar x_1) + p(1 - \bar x_2) - \min\{p^2, p/2 \} (1 - \bar x_1) (1 - \bar x_2) \bigr)
$$

At the end of the process, $P' =\emptyset$ and hence $\Phi' = \bar x'_1 \bar x'_2$. The final independent rounding step, conditioned on the terminal state $ x'$, then gives $\E[ X_{\bar e_1} X_{\bar e_2} ] = \bar x'_1 \bar x'_2$.
\end{proof}

This already gives significantly improved constant-factor guarantees compared to \cite{baveja2024}:

\begin{corollary}
With $p = 1/3$, distinct edges $e,f$ sharing a left-node have $\E[X_e X_f] \leq \tfrac{23}{27} x_e x_f \leq 0.852 \  x_e x_f$.
\end{corollary}
\begin{proof}
Since $$
\bigl( 1 -p + p^2 ( 2 -  x_{1} - x_{2})  - \min\{ p^3, p^2/2 \} (1 - x_{1})(1 - x_{2})  \bigr) \leq 23/27
$$
for all $x_1, x_2 \in [0,1]$ with $x_1 + x_2 \leq 1$.
\end{proof}

To improve further, we will combine this Brownian-motion algorithm with the contention-resolution algorithm.

\section{Combining algorithms}
\label{combine-algorithm}
At this point, we recall the alternative rounding algorithm of \cite{harris2025}. It satisfies properties (A1), (A2), (A3), and it also satisfies a strong-negative-correlation property similar to (A4). The full result is quite technical to state; we summarize it as follows:
\begin{theorem}[\cite{harris2025}]
Suppose we have a vector $\rho \in [0,1]^E$ with $\rho(N(u)) = 1$  for each left-node $u$.  Then, for any distinct edges $e,f$ sharing a common left-node, the Dirichlet rounding mechanism satisfies
$$
\E[X_e X_f] \leq x_e x_f \Psi(x_e,x_f;\rho_e, \rho_f)
$$
for an explicitly-given function $\Psi$ which satisfies the following bound:
$$
\Psi(x_1, x_2; \rho_1, \rho_2) \leq \frac{\Gamma(\rho_1 (1/x_1 - 1) + 1) \Gamma(\rho_2 (1/x_2 - 1) + 1)}{ \Gamma(\rho_1 (1/x_1-1) + \rho_2 (1/x_2 - 1) + 1)}
$$
\end{theorem}

\begin{observation}
\label{together-obs}
For any parameters $p,q \in [0,1]$, we can achieve properties (A1), (A2), (A3), and distinct edges $e,f$ sharing a common left-node satisfy
$$
\E[X_e X_f] \leq x_e x_f \Bigl( q \bigl( 1-p + p^2 (2 - x_e - x_f) - \min\{p^3, p^2/2 \} (1 - x_e)(1 - x_f) \bigr) + \frac{(1-q) \Gamma(2 - x_f) \Gamma(2 - x_e)}{\Gamma(3 - x_e - x_f)}  \Bigr)
$$
\end{observation}
\begin{proof}
Since we are given a fractional matching $\vec x$, we will naturally set $\rho_e = x_e$ for all edges $e$. Then, the Dirichlet mechanism gives
$$
\E[X_e X_f ] \leq \frac{x_e x_f \Gamma(2-x_e) \Gamma(2-x_f)}{\Gamma(3-x_e -x_f)} \
$$

With probability $q$, we run \textsc{DepRoundBrownian} and with probability $1 - q$ we run the Dirichlet rounding algorithm.
\end{proof}

\begin{proposition}
\label{tatat1}
Let $x_1, x_2 \geq 0$ with $x_1 + x_2 \leq 1$. For $p = 0.439327, q = 0.822077$, there holds:
$$
q \bigl( 1-p + p^2 (2 - x_1 - x_2) - \min\{p^3, p^2/2 \} (1 - x_1)(1 - x_2) \bigr) +  \frac{(1-q) \Gamma(2-x_1) \Gamma(2-x_2)}{\Gamma(3-x_1-x_2)}  \leq 0.79751
$$
\end{proposition}
\begin{proof}
Let $t = x_1 + x_2$. By straightforward calculus, we can observe that
$$
p^2 (2 - x_1 - x_2) - \min\{p^3, p^2/2 \} (1 - x_1)(1 - x_2)  \leq p^2 (2 - t) - \min\{p^3, p^2/2 \} (1 - t)
$$
and
$$
\frac{\Gamma(2-x_1)\Gamma(2-x_2)}{\Gamma(3-x_1-x_2)} \leq  \frac{1}{2-t}.
$$
So it suffices to show that, for $t \in [0,1]$, we have
\begin{equation}
\label{rar112}
q (1-p + p^2 (2 - t) - \min\{p^3, p^2/2 \} (1 - t) ) + \frac{1-q}{2-t} \leq 0.79751
\end{equation}

The LHS is an algebraic function of $t$, so it can be shown mechanically.
\end{proof}

Putting together \Cref{together-obs} and \Cref{tatat1}, this concludes the proof of \Cref{mainthm34}.

\section{Lower bound}
\label{lb}

Consider graph $G=C_{30}$ with vertices
$u_1,v_1,u_2,v_2,\ldots,u_{15},v_{15}$ in cyclic order
(left-nodes $u_i$, right-nodes $v_i$), and $x_e=\tfrac12$ on every edge.
Define $Z_i=X_{u_i,v_i}$ and
$\bar Z_i=X_{u_{i+1},v_i}=1-Z_i$, where indices are taken modulo~$15$.

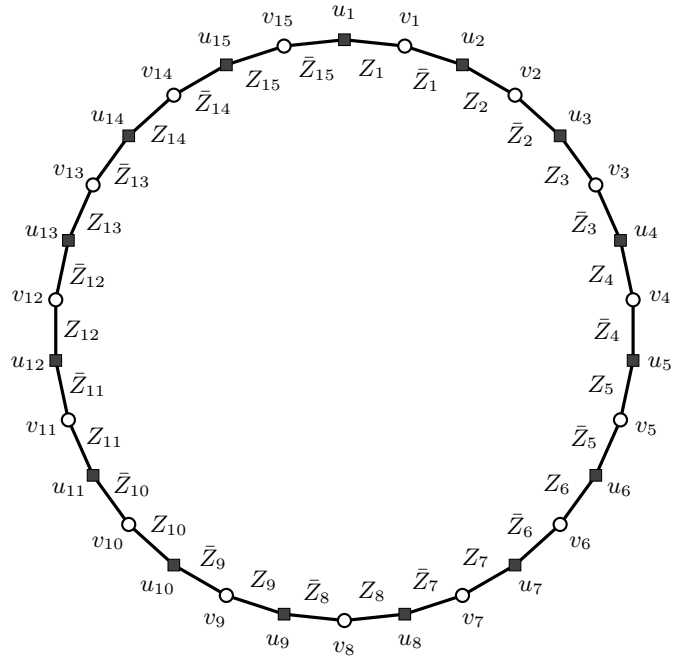
\begin{figure}[h]
\centering
\begin{tikzpicture}[scale=0.96]
  % coordinates
  \foreach \i in {1,...,15}{
    \pgfmathsetmacro{\tu}{90-24*(\i-1)}
    \pgfmathsetmacro{\tv}{78-24*(\i-1)}
    \coordinate (u\i) at (\tu:4.0);
    \coordinate (v\i) at (\tv:4.0);
  }

  % edges and labels on a single inner ring
  \foreach \i in {1,...,15}{
    \pgfmathsetmacro{\ma}{84-24*(\i-1)}
    \pgfmathsetmacro{\mb}{72-24*(\i-1)}
    \pgfmathtruncatemacro{\ip}{mod(\i,15)+1}

    \draw[very thick] (u\i) -- (v\i);
    \draw[very thick] (v\i) -- (u\ip);

    \node[font=\footnotesize] at (\ma:3.62) {$Z_{\i}$};
    \node[font=\footnotesize] at (\mb:3.62) {$\bar Z_{\i}$};
  }

  % vertices and vertex labels
  \foreach \i in {1,...,15}{
    \pgfmathsetmacro{\tu}{90-24*(\i-1)}
    \pgfmathsetmacro{\tv}{78-24*(\i-1)}

    \filldraw[fill=black!75]
      (u\i) +(-0.080,-0.080) rectangle +(0.080,0.080);
    \filldraw[fill=white, draw=black, thick]
      (v\i) circle (0.090);

    \node[font=\footnotesize] at (\tu:4.38) {$u_{\i}$};
    \node[font=\footnotesize] at (\tv:4.38) {$v_{\i}$};
  }
\end{tikzpicture}
\end{figure}

We record the following combinatorial fact  (constructed by solving a dual linear program).

\begin{observation}
\label{simple-obs}
Define
\begin{align*}
Q_d&=\sum_{i=1}^{15}Z_i\bar Z_{i-d}
\end{align*}
where indices are interpreted mod 15.

Then $$
16Q_1-14Q_2+11Q_3-9Q_4+6Q_5-4Q_6+Q_7
\geq -3.
$$
\end{observation}
\begin{proof}
Try all $2^{15}$ possible sequences $(Z_1, \dots, Z_{15})$. 
\end{proof}

\begin{theorem}
Consider any bipartite rounding algorithm satisfying  properties (A1), (A2) for $G$, which also satisfies the modified bounds
\begin{itemize}
\item[(A3-)] $\E[ X_e X_f ] \leq x_e x_f$ for any pair of edges $e,f$ whose distance in the line graph is at least 3.

\item[(A4')] $\E[ X_e X_f ] \leq c \ x_e x_f$ for any distinct  edges $e,f$ sharing a left-node. 
\end{itemize}

Then $c \geq 41/80$.
\end{theorem}

\begin{proof}
Each of the fifteen summands $Z_i\bar Z_{i-1}$ in $Q_1$ corresponds to the
two edges $(u_i,v_i)$ and $(u_i,v_{i-1})$, which share left-node $u_i$ and have $x$-values of $1/2$.
Therefore
\[
\E[Q_1]\leq\frac{15c}{4}.
\]

We next claim that, for any $i \in \{1, \dots, 15 \}, d \in \{2, \dots, 13 \}$, we have
\begin{equation}
\label{yw1}
\E[ Z_i Z_{i-d} ] = \E[ Z_i \bar Z_{i-d} ] = \E[ \bar Z_i Z_{i-d} ]  = \E[ \bar Z_i \bar Z_{i - d} ] = 1/4
\end{equation}

For, all four of the corresponding edge pairs have distance at least 3. So, property (A3-) gives
\begin{equation}
\label{yw2}
\E[ Z_i Z_{i-d} ], \E[ Z_i \bar Z_{i-d} ],  \E[ \bar Z_i Z_{i-d} ], \E[ \bar Z_i \bar Z_{i - d} ] \leq 1/4
\end{equation}

On the other hand, the four random variables \[
Z_iZ_{i-d},\qquad
Z_i\bar Z_{i-d},\qquad
\bar Z_iZ_{i-d},\qquad
\bar Z_i\bar Z_{i-d}
\] are nonnegative and sum to one. So all the inequalities in (\ref{yw2}) must be tight. 

By (\ref{yw1}), we thus have
$$
\E[Q_d]=15/4 \qquad \text{for $d\in\{2,\ldots,13\}$}
$$

Taking expectations in \Cref{simple-obs}, we obtain\[
-3
\leq
16 (15 c/4)
+
(-14+11-9+6-4+1) (15/4)
=
60c- 135/4
\]
and consequently
\[
c\geq\frac{41}{80}. \qedhere
\]
\end{proof}

\section{Acknowledgments}
Thanks to Karthik Abinav, for bringing the $C_{30}$ lower bound construction (found via ChatGPT) to my attention.

\bibliographystyle{alpha}
\bibliography{refs}

\begin{thebibliography}{HLRV26}

\bibitem[BQS24]{baveja2024}
Alok Baveja, Xiaoran Qu, and Aravind Srinivasan.
\newblock Approximating weighted completion time via stronger negative correlation.
\newblock {\em Journal of Scheduling}, 27(4):319--328, 2024.

\bibitem[BSS21]{bansal2016lift}
Nikhil Bansal, Aravind Srinivasan, and Ola Svensson.
\newblock Lift-and-round to improve weighted completion time on unrelated machines.
\newblock {\em SIAM Journal on Computing}, 50(3):STOC16--138, 2021.

\bibitem[Har25]{harris2024dependent}
David~G. Harris.
\newblock Dependent rounding with strong negative-correlation, and scheduling on unrelated machines to minimize completion time.
\newblock {\em ACM Transactions on Algorithms}, 22:13:1--13:24, 2025.

\bibitem[HLRV26]{harris2025}
David~G. Harris, George~Z. Li, Nitya Raju, and Renata Valieva.
\newblock The {D}irichlet {M}echanism for rounding with strong negative correlation, with applications.
\newblock In {\em 53rd International Colloquium on Automata, Languages, and Programming (ICALP)}, pages 107:1--107:21, 2026.

\bibitem[IL23]{im2023}
Sungjin Im and Shi Li.
\newblock Improved approximations for unrelated machine scheduling.
\newblock In {\em Proc. 2023 ACM-SIAM Symposium on Discrete Algorithms (SODA)}, pages 2917--2946, 2023.

\bibitem[IS20]{im2020}
Sungjin Im and Maryam Shadloo.
\newblock Weighted completion time minimization for unrelated machines via iterative fair contention resolution.
\newblock In {\em Proc. 2020 ACM-SIAM Symposium on Discrete Algorithms (SODA)}, pages 2790--2809, 2020.

\end{thebibliography}
\end{document}